\documentclass{article}
\usepackage{ijcai17}
\usepackage{times}
\usepackage{url}
\usepackage{amsmath}
\usepackage{amsthm}
\usepackage{graphicx}
\usepackage[small]{caption}
\usepackage{graphics}
\usepackage{epsfig}
\usepackage{latexsym}
\usepackage{amsfonts}
\usepackage{paralist}
\usepackage{xspace}
\usepackage{setspace}
\usepackage{mathrsfs}
\usepackage{amssymb}
\usepackage{color}
\usepackage{algorithmic}
\usepackage[ruled]{algorithm}
 \usepackage{bm}
\newtheorem{theorem}{Theorem}

\newtheorem{definition}{Definition}
\newtheorem{lemma}{Lemma}

\title{Toward Optimal Coupon Allocation in Social Networks: An Approximate Submodular Optimization Approach}
\author{Shaojie Tang\\
Naveen Jindal School of Management\\
University of Texas at Dallas\\
shaojie.tang@utdallas.edu\\
}

\begin{document}

\maketitle

\begin{abstract}
CMO Council reports that 71\% of internet users in the U.S. were influenced by coupons and discounts when making their purchase decisions. It has also been shown that offering coupons to a small fraction of users (called seed
users) may affect the purchase decisions of many other users in a social network. This motivates us to study the optimal coupon allocation problem, and our objective is to allocate coupons to a set of users so as to maximize the expected cascade. Different from existing studies on influence maximizaton (IM), our framework allows a general utility function and a more complex set of
constraints. In particular, we formulate our problem as an approximate submodular maximization problem subject to matroid and knapsack constraints. Existing techniques relying on the submodularity of the utility function, such as greedy algorithm, can not work directly on a non-submodular function. We use $\epsilon$ to measure the difference between our function and its closest submodular function and propose a novel approximate algorithm with approximation ratio $\beta(\epsilon)$ with $\lim_{\epsilon\rightarrow 0}\beta(\epsilon)=1-1/e$. This is the best approximation guarantee for approximate submodular maximization subject to a partition matroid and knapsack constraints, our results apply to a broad range of optimization problems that can be formulated as an approximate submodular maximization problem.
\end{abstract}

%%%%%%%%%%%%%%%%%%%%%%%%%%%%%%%%%%%%%%%%%%%%%%%%%%%%%%%%%%%%%%%%%%%%%%

% Samples of sectioning (and labeling) in ISRE
% NOTE: (1) \section and \subsection do NOT end with a period
%       (2) \subsubsection and lower need end punctuation
%       (3) capitalization is as shown (title style).
%
%\section{Introduction.}\label{intro} %%1.
%\subsection{Duality and the Classical EOQ Problem.}\label{class-EOQ} %% 1.1.
%\subsection{Outline.}\label{outline1} %% 1.2.
%\subsubsection{Cyclic Schedules for the General Deterministic SMDP.}
%  \label{cyclic-schedules} %% 1.2.1
%\section{Problem Description.}\label{problemdescription} %% 2.

%%%%%%%%%%%%%%%%%%%%%%%%%%%%%%%%%%%%%%%%%%%%%
\section{Introduction}
\label{sec:intro}
%Newsfeed is a constantly updating list of stories from people and pages that you follow on Facebook. Newsfeed stories include status updates, photos, videos, links, app activity and likes. As the user scrolls, it is possible to insert sponsored/promoted posts in the newsfeed on the user screen.
It has been reported that one of the most effective ways of running promotions on social media is the use of coupon campaigns on the largest social media site such as Facebook and Twitter.  Different from conventional online advertising, the offer of a coupon on social network could trigger a large amount of likes and shares on your product and build brand awareness rapidly.  This motivates us to study the coupon allocation problem, e.g., allocating coupons to a few users so as to maximize the expected cascade. More specifically, we assume that the company would like to promote a product through offering coupons with different values to a subset of users. Whether a user would accept the coupon and further become the seed node is probabilistic, which depends on the value of the coupon. Our goal is to select a set of users and allocate one proper coupon to each of them, such that the expected cascade is maximized.

We notice that our problem is closely related to influence maximization problem (IM). Although IM has been extensively studied in the literature \cite{kempe2003maximizing,chen2013information,leskovec2007cost,cohen2014sketch,chalermsook2015social}, most of existing studies assume a ``budgeted, deterministic and submodular'' setting, where there
is a predefined budget for selecting seed nodes, any node is guaranteed to be activated as a seed
node once it is selected, and the utility function is assumed to be submodular in terms of seed nodes. However, these assumptions may not always hold in reality. Firstly, the probability that a user is activated as a seed node depends on many factors, including the actual amount of rewards received from the company \cite{yangcontinuous}. Secondly, the utility function may not always satisfy the submodular property, e.g., some user can only be influenced if majority of her friends are activated \cite{kuhnle2017scalable}.

 In this paper, we study the coupon allocation problem that allows a general utility function. In particular, we assume an $\epsilon$-approximate submodular function where $\epsilon$ measures the distance between a given function and its closest submodular function (a detailed definition can be found in Section \ref{sec:pre}). We formulate the coupon allocation problem as an approximate submodular maximization problem subject to a partition matroid and knapsack constraints. Existing techniques relying on the submodularity of the utility function, such as greedy algorithm, can not work directly on a non-submodular function. We propose a novel approximate algorithm with approximation ratio $\beta(\epsilon)$ with $\lim_{\epsilon\rightarrow 0}\beta(\epsilon)=1-1/e$. This is, to the best of our knowledge, the strongest theoretical result available for approximate submodular maximization problem subject to a partition matroid and knapsack constraints. Although we restrict our attention to coupon allocation in this paper, our results apply to a broad range of non-submodular maximization problems.

The contributions of this paper can be summarized as follows:
(1) We are the first to study the coupon allocation problem that allows a general utility function. Our objective is to determine the best coupon allocation so as to maximize the expected cascade subject to (partition) matroid and knapsack constraints.
(2) We develop an efficient algorithm with approximation ratio that depends on $\epsilon$. When $\epsilon=0$ (the utility function is submodular), our result converges to $1-1/e$ which is the best possible for submodular maximization subject to matroid constraints. This research contributes fundamentally to the development of approximate solutions for any problems that fall into the family of approximate submodular maximization.

Some important notations are listed in Table \ref{symbol}.

\begin{table}[t]\centering
\caption{Symbol table.}
\begin{tabular}{{|c|l|}}
\hline
\textbf{Notation} & \textbf{Meaning}\\
\hline
\hline
$\mathcal{S}$ & the ground set that contains all u-c pairs\\
\hline
$\mathcal{S}_v$ & all u-c pairs whose user is $v$\\
\hline
$f(S)$ & the expected cascade of $S$\\
\hline
$\mathbf{y}$ & a decision matrix\\
\hline
$F(\mathbf{y})$ & the multilinear extension of $f$\\
\hline
$f^+(\mathbf{y})$ & the concave extension of $f$\\
\hline
$g$ & a submodular function used to bound $f$\\
\hline
$\mathbf{a}\oplus \mathbf{b}$ & the coordinate wise maximum of $\mathbf{a}$ and $\mathbf{b}$ \\
\hline
$F_\mathbf{\mathbf{x}}([vd])$ & $F(\mathbf{x}\oplus \mathbf{1}_{vd})-F(\mathbf{x})$\\
%\hline
%$\alpha(t); \beta(t)$ & slot decay function\\
%\hline
%$\pi[t]$ & the ad placed in slot $t$ by $\pi$\\
%\hline
%$k_{\pi[t]}$ & the exposure index of ad $\pi[t]$\\
%\hline
%$r_i$ & the release time of $J_i$ (used in online setting)\\
%\hline
%$o_i$ & the deadline of job $J_i$\\
\hline
\end{tabular}
\label{symbol}
\end{table}

\section{Related Work}
IM has been extensively studied in the literature \cite{kempe2003maximizing,chen2013information,leskovec2007cost,cohen2014sketch,chalermsook2015social}, their objective is to find a set of influential customers so as to maximize the expected cascade. However, our work differ from all existing studies in several major aspects. Traditional IM assumes any node is guaranteed to be activated once it is selected, we relax this assumption by allowing users to response differently to different coupon values.  Recently, \cite{yangcontinuous} study discount allocation problem in social networks. However, they assume a submodular utility function and continuous setting of discount value, our model allows a general utility function. We formulate our problem as an approximate submodular maximization problem subject to matroid and knapsack constraints.  Existing approaches, such as greedy algorithm \cite{horel2016maximization}, can not apply directly to matroid and knapsack constraints. We propose a novel algorithm that provides the first bounded approximate solutions to this problem. It was worth noting that our approximation ratio converges to $1-1/e$ when $\epsilon=0$. This is the best theoretical result available for approximate submodular maximization subject to a partition matroid constraint.

\section{Preliminaries}
\subsection{Submodular Function and Its Continuous Extensions}
\label{sec:pre}
A submodular function is a set function $h:2^{\Omega}\rightarrow \mathbb{R}$, where $2^{\Omega } $ denotes the power set of  $\Omega$, which satisfies a natural ``diminishing returns" property: the marginal gain from adding an element to a set $X$ is at least as high as the marginal gain from adding the same element to a superset of $X$. Formally, a submodular function satisfies the follow property: For every $X, Y \subseteq \Omega$ with $X \subseteq Y$ and every $x \in \Omega \backslash Y$, we have that $h(X\cup \{x\})-h(X)\geq h(Y\cup \{x\})-h(Y)$
We say a submodular function $f$ is monotone if $h(X) \leq h(Y)$ whenever $X \subseteq Y$.

\subsubsection{Continuous Extensions} Consider any vector $\bold{x}=\{x_1,x_2,\ldots,x_n\}$ such that each $0\leq x_i\leq 1$. The multilinear extension of $h$ is defined as $H(\bold{x})=\sum_{X\subseteq \Omega} h(X) \prod_{i\in X} x_i \prod_{i\notin X} (1-x_i)$, and the concave extension of $h$ is defined as  $
h^+(\mathbf{x})=\max \{\sum_{X\subseteq \Omega} \alpha_X f(X)| \alpha_X \geq 0, \sum_{X\subseteq \Omega}\alpha_X \leq1, \sum_{X}\alpha_X\mathbf{1}_X \leq \mathbf{x}\}$.

\subsubsection {$\epsilon$-approximate Submodular Function} An  $\epsilon$-approximate submodular function is a set function $f:2^{\Omega}\rightarrow \mathbb{R}$ that satisfies the following condition: there exists a submodular function $h$ such that for any $X\subseteq \Omega$, we have $(1-\epsilon)h(X) \leq f(X) \leq (1+\epsilon)h(X)$.
%We consider the problem of  selecting and displaying  a set of promoted posts  to display in an ad space with $m$ slots. Each ad $a_i \in \mathcal{A}$ is associated with a per-click revenue $\lambda_{a_i}$, which specifies the amount of advertising revenue that is generated from a click on that ad. In order to  maximize the expected revenue that is generated, the ad network needs to optimize their selection of ads and carefully decide the sequencing of selected ads in an ad space. Next we introduce the click-through-model and important notations adopted in this paper.

\subsection{Coupon Adoption and Propagation Model}
Assume there are $n$ users $\mathcal{V}=\{1, 2,\dots, n\}$ and $m$ types of coupons $\mathcal{D}=\{1, 2,\dots, m\}$. For simplicity of presentation, we will directly use $d$ to denote the coupon value of $d$.  Given the set of users $\mathcal{V}$ and possible coupon values $\mathcal{D}$, define $\mathcal{S} \triangleq \mathcal{V}\times \mathcal{D}$ as the solution space, adding a \emph{user-coupon} (u-c) pair $[vd] \in \mathcal{S} $ to our solution translates to offering coupon  $d\in \mathcal{D}$ to user $v\in \mathcal{V}$. When any coupon $d\in \mathcal{D}$ is allocated to any user $v\in \mathcal{V}$, we assume that with probability $p_v(d)\in(0,1]$, $v$ adopts $d$ and becomes a seed node, thus incurs a cost $d$. We further assume that coupons cannot be combined, if a user received multiple coupons, her adoption decision only depends on the coupon with the highest value.

After a set of users become seed nodes, they start to influence other users. Assume $U$ is the seed set, the expected cascade of $U$, which is the expected number of influenced users given seed set $U$, is denoted as $\gamma(U)$. In this work, we assume that $\gamma(U)$ is a $\epsilon$-approximate submodular function: there exists a submodular function $q$ such that for any $U\in \mathcal{V}$, we have $(1-\epsilon)q(U) \leq \gamma(U) \leq (1+\epsilon)q(U)$. Our propagation model subsumes many classic propagation models, including \emph{Independent Cascade Model} and \emph{Linear Threshold Model} \cite{kempe2003maximizing} as special cases, as the propagation functions defined in their models are submodular and monotone, i.e., $\epsilon=0$.

\section{Problem Statement}
Our objective is to identify the best coupon allocation policy, not necessarily deterministic, to maximize the expected cascade.
 Given a coupon allocation $S\subseteq \mathcal{S}$, we use $d_S(v)$ to denote the largest coupon value allocated to user $v$ under $S$, then the probability that a subset of users $U \subseteq \mathcal{V}$ successfully become the seed set is
\[\Pr(U;S)=\prod_{u\in U} p_u(d_S(u)) \prod_{v\in \mathcal{V}\setminus U}(1- p_v (d_S(v)))\]
As introduced earlier, we use $\gamma(U)$ to denote the expected cascade under seed set $U$, then the expected cascade under allocation $S$ is
 $f(S)=\sum_{U\in 2^{\mathcal{V}}}\Pr(U;S) \gamma(U)$.

 We first prove that $f$ is also $\epsilon$-approximate submodular.
\begin{lemma}
\label{lem:00}
If $\gamma$ is $\epsilon$-approximate submodular, $f$ is also $\epsilon$-approximate submodular: there exists a submodular function $g$ such that  $\forall S\in \mathcal{S}: (1-\epsilon)g(S) \leq f(S) \leq (1+\epsilon)g(S)$.
\end{lemma}
\begin{proof}Based on a similar proof provided in \cite{soma2015generalization}, we can show that if  $q(U)$ is submodular of $U$, $\sum_{U\in 2^{\mathcal{V}}}\Pr(U;S) q(U)$ is submodular of $S$. Because \[(1-\epsilon)q(U) \leq \gamma(U) \leq (1+\epsilon)q(U)\] we have $(1-\epsilon)\sum_{U\in 2^{\mathcal{V}}}\Pr(U;S) q(U) \leq f(S)\leq (1+\epsilon)\sum_{U\in 2^{\mathcal{V}}}\Pr(U;S) q(U)$. We finish the proof by defining $g(S)=\sum_{U\in 2^{\mathcal{V}}}\Pr(U;S) q(U)$.
\end{proof}
%Define $g(S)=\sum_{U\in 2^{\mathcal{V}}}\Pr(U;S) q(U)$, Lemma \ref{lem:00} implies that
%$\forall S\in \mathcal{S}: (1-\epsilon)g(S) \leq f(S) \leq (1+\epsilon)g(S)$.

The expected cost of a coupon allocation $S$ is
\[c(S)=\sum_{U\in 2^{\mathcal{V}}} \left(\Pr(U;S)\sum_{u\in U}d_S(u)\right)\]

\paragraph{Coupon Allocation Policy} We denote by $\Theta=\{\theta_S:S\subseteq \mathcal{S}\}$ a coupon allocation policy, where $\theta_S$ is the probability that $S$ is adopted.
Then the expected cascade under $\Theta$ is
$f(\Theta)=\sum_{S\in \mathcal{S}}\theta_{S} f(S)$. The expected cost of $\Theta$ is
$c(\Theta)= \sum_{S\in \mathcal{S}}\theta_{S} c(S)$.

\begin{definition}
[Feasible Coupon Allocation Policy] We say a policy $\Theta$ is feasible if and only the following conditions are satisfied:
\begin{itemize}
\item(Attention Constraint) For every user $v\in \mathcal{V}$, we denote $\mathcal{S}_v=\{[vd]: d\in \mathcal{D}\}$. $\forall \theta_S>0, \forall v: |\mathcal{S}_v\cap S|\leq 1$, e.g., every user receives at most one coupon.
\item(Budget Constraint) $ \sum_{S\in \mathcal{S}}\theta_{S} c(S)\leq B$, e.g., the expected  value of the coupons redeemed is at most $B$.
\end{itemize}\label{def:0}
\end{definition}

 %\begin{center}
%\framebox[0.5\textwidth][c]{
%\enspace
%\begin{minipage}[t]{0.5\textwidth}
%\small
%\textbf{P.A} $\mathbb{E}_{\theta_{\mathcal{S}}\sim \Theta} [f(S)]$\\
%\textbf{subject to:}
%$\Theta$ is feasible
%%$\forall \theta_{\mathcal{S}}> 0: |\mathcal{S}|\leq 1 $\\
%%$\mathbb{E}_{\theta_{\mathcal{S}}\sim \Theta} [c(\mathcal{S})]\leq B$\\
%%$\sum_{\mathcal{S}\in s}\theta_{\mathcal{S}}= 1$
%\end{minipage}
%}
%\end{center}
%\vspace{0.1in}

Our objective is to identify a feasible coupon allocation policy that maximizes the expected cascade. We present the formal definition of our problem in \textbf{P.A}.
 \begin{center}
\framebox[0.45\textwidth][c]{
\enspace
\begin{minipage}[t]{0.45\textwidth}
\small
\textbf{P.A} $\max_{\Theta} \sum_{S\in \mathcal{S}}\theta_{S} f(S)$\\
\textbf{subject to:}
\begin{equation*}
\begin{cases}
\forall \theta_S>0, \forall v\in \mathcal{V}: |\mathcal{S}_v\cap S|\leq 1 \quad \mbox{(Ca)}\\
%\forall \theta_S>0: |S| \leq K\\
\sum_{S\in \mathcal{S}}\theta_{S} c(S)\leq B \quad \mbox{(Cb)}\\
\sum_{S\in \mathcal{S}}\theta_{S}= 1 \quad \mbox{(Cc)}
\end{cases}
\end{equation*}
\end{minipage}
}
\end{center}
\vspace{0.1in}

We extend our model in Section \ref{sec:extended} by incorporating one more constraint, e.g., a feasible policy can not allocate coupons to more than $K$ users. This constraint captures the fact that the company often has limited budgeted on coupon producing and distribution.

\section{An Approximate Solution}

We first introduce a new problem \textbf{P.B} as follows.
 \begin{center}
\framebox[0.45\textwidth][c]{
\enspace
\begin{minipage}[t]{0.45\textwidth}
\small
\textbf{P.B:}
\emph{Maximize $f^+(\mathbf{y})$}\\
\textbf{subject to:}
\begin{equation*}
\begin{cases}
%$\forall \theta_{\mathcal{S}}> 0: |\mathcal{S}|\leq 1 $\\
\forall v\in \mathcal{V}: \sum_{d\in \mathcal{D}} y_{vd}\leq1 \quad\mbox{(C1)}\\
% \sum_{v\in V}\sum_{d\in D} y_{vd}\leq K \quad(C2)\\
  \sum_{[vd]\in \mathcal{S}} y_{vd}p_v(d)d\leq B \quad\mbox{(C2)}\\
\forall [vd]\in \mathcal{S}: y_{vd} \in[0,1] \quad \mbox{(C3)}
\end{cases}
\end{equation*}
\end{minipage}
}
\end{center}
\vspace{0.1in}
In the above formulation,  $\mathbf{y}$ is a $n\times m$  decision matrix  and $f^+(\mathbf{y})$ is a concave extension of $f$.

\begin{equation}
f^+(\mathbf{y}) = \max \left\{\sum_{S\in \mathcal{S}} \alpha_S f(S) \middle\vert \begin{array}{l}
    \alpha_S \geq 0;\\
      \sum_{S\in \mathcal{S}}\alpha_S\leq1;\\
      \forall v: \sum_{S \ni [vd] }\alpha_S \leq y_{vd}
  \end{array}\right\}
  \label{eq:1}
\end{equation}
%\begin{eqnarray}f^+(\mathbf{y})=\max \{\sum_{S\in \mathcal{S}} \zeta_S f(S)|
%%$\forall \theta_{\mathcal{S}}> 0: |\mathcal{S}|\leq 1 $\\
%\zeta_S \geq 0;
%% \sum_{v\in V}\sum_{d\in D} y_{vd}\leq K \quad(C2)\\
%  \sum_{S\in \mathcal{S}}\zeta_S\leq1;\\
%\forall v\in \mathcal{V}: \sum_{S \ni [vd] }\zeta_S \leq y_{vd}
%\}\end{eqnarray}

 We first prove that  \textbf{P.B} is a relaxed version of \textbf{P.A}.
\begin{lemma} Assume $\Theta^*$ is the optimal solution to \textbf{P.A} and $\mathbf{y}^+$ is the optimal solution to \textbf{P.B}, we have
$f^+(\mathbf{y}^+)\geq \sum_{S\in \mathcal{S}}\theta^*_{S} f(S)$.
\end{lemma}
\begin{proof}
Given the optimal policy $\Theta^*$, for every u-c pair $[vd]\in \mathcal{S}$, we define $y^*_{vd}$ as the probability that $[vd]$ is offered by $\Theta^*$, e.g., $y^*_{vd}=\sum_{S\ni [vd]}\theta^*_S$. We first prove that $\mathbf{y}^*$ is a feasible solution to \textbf{P.B}.
Because $\Theta^*$ is a feasible policy, we have
\begin{enumerate}
\item $\forall \theta^*_S>0: |\mathcal{S}_v\cap S|\leq 1$, e.g., every user receives at most one coupon under $\Theta^*$. It follows that $\forall v\in \mathcal{V}:  \sum_{d\in \mathcal{D}} y^*_{vd}= \sum_{d\in \mathcal{D}}\sum_{S\ni [vd]}\theta^*_S\leq \sum_{S\in \mathcal{S}} \theta^*_S=1$, thus condition (C1) is satisfied.

%(2) $\forall \theta^*_S>0: |S| \leq K$, e.g., the number of users who receive a coupon under $\Theta^*$ is no larger than $K$. Since $\sum_{v\in V}\sum_{d\in D} y^*_{vd}$ is the expected number of users who receive coupons under $\Theta^*$, we have $\sum_{v\in V}\sum_{d\in D} y^*_{vd}=\sum_{v\in V}\sum_{d\in D} \sum_{S\ni [vd]}\theta^*_S\leq K$, thus condition (C2) is satisfied.

\item $\sum_{S\in \mathcal{S}}\theta^*_{S} c(S)\leq B$, e.g., the expected cost of any coupon allocation under  $\Theta^*$ is no larger than $B$. Because $\sum_{[vd]\in \mathcal{S}} y^*_{vd}p_v(d)d$ is the expected cost of $\Theta^*$, we have $\sum_{[vd]\in \mathcal{S}} y^*_{vd}p_v(d)d\leq B$,  thus condition (C2) is satisfied.

\item  $\sum_{S\in \mathcal{S}}\theta_{S}= 1$. Because $\forall [vd]\in \mathcal{S}: y_{vd}$ is the probability that $[vd]$ is offered by $\Theta^*$, we have $\forall [vd]\in \mathcal{S}: y_{vd} \in[0,1]$,  thus condition (C3) is satisfied.
\end{enumerate}
On the other hand, by setting $\alpha_S=\theta_S$ for every $S\in \mathcal{S}$ in (\ref{eq:1}), we have $f^+(\mathbf{y}^*)\geq \sum_{S\in \mathcal{S}}\theta^*_{S} f(S)$.
\end{proof}
Assume $g^+$ is the concave extension of $g$, we next prove that $g^+(\mathbf{y})$ is an approximate of $f^+(\mathbf{y})$.
\begin{lemma}
\label{lem:3}
$f^+(\mathbf{y})\leq (1+\epsilon)g^+(\mathbf{y})$
\end{lemma}
\begin{proof}Given any $\mathbf{y}$, assume the value of $f^+(\mathbf{y})$ is achieved at $\{\alpha^*_A|A\subseteq \mathcal{V}\}$, we have $\sum_{A\subseteq \mathcal{V}} \alpha^*_A f(A)\leq \sum_{A\subseteq \mathcal{V}} \alpha^*_A (1+\epsilon)g(A)=(1+\epsilon)\sum_{A\subseteq \mathcal{V}} \alpha^*_A g(A)\leq (1+\epsilon)g^+(\mathbf{y})$.
\end{proof}

\subsection{Algorithm Design}
In this section, we present a greedy algorithm that achieves a bounded approximation ratio. Our general idea is to first find a fractional solution with a bounded approximation ratio and then round it to an integral solution.
\subsubsection{Continuous Greedy}

\begin{algorithm}[h]
{\small
\caption{Continuous Greedy}
\label{alg:greedy-peak}
%\textbf{Input:} Social network $\mathcal{G}$, budget $\mathcal{B}$, individual attention constraint $\kappa_i$, overall attention constraint $K$.\\
%\textbf{Output:} Seed set $\mathcal{S}$.
\begin{algorithmic}[1]
\STATE Set $\delta=1/(nm)^2, t=0, f(\emptyset)=0$.
\WHILE{$t<1$}
\STATE Let $R(t)$ contain each $[vd]$ independent with probability $y_{vd}(t)$.
\STATE For each $[vd]\in \mathcal{S}$, estimate
\[\omega_{vd}=\mathbb{E}[f(R(t)\cup\{[vd]\})]-\mathbb{E}[[f(R(t))]\]
\STATE Solve the following liner programming problem and obtain the optimal solution $\overline{\mathbf{y}}$
\STATE
\framebox[0.42\textwidth][c]{
\enspace
\begin{minipage}[t]{0.42\textwidth}
\small
\textbf{P.C:}
\emph{Maximize $\sum_{[vd]\in \mathcal{S}}\omega_{vd}y_{vd}$ }\\
\textbf{subject to:} Conditions (C1)$\sim$(C3)
\end{minipage}
}
\vspace{0.1in}
\STATE Let $y_{vd}(t+\delta)=y_{vd}(t)+\overline{y}_{vd}$; \label{line:1}
\STATE Increment $t=t+\delta$;
\ENDWHILE
\end{algorithmic}
}
\end{algorithm}

In \cite{vondrak2008optimal} they develop a continuous greedy algorithm based on the multilinear extension in order to maximize a submodular
monotone function over a matroid constraint. We extend their results to non-sumbodular maximization subject to a partition matroid and knapsack constraints. We use $\mathbf{1}_{vd}$ to denote a $n\times m$ matrix with entry $(v,d)$ one and all other entries zero. Given two matrices $\mathbf{a}$ and $\mathbf{b}$, let $\mathbf{a}\oplus \mathbf{b}$ denote the coordinate-wise maximum.  Define $F_\mathbf{\mathbf{x}}([vd])=F(\mathbf{x}\oplus \mathbf{1}_{vd})-F(\mathbf{x})$. A detailed description of our algorithm is listed in Algorithm \ref{alg:greedy-peak}.

%\[\frac{1-\epsilon}{1+\epsilon}f^+(\mathbf{y}^+)-(1+ \frac{2\epsilon K}{1+\epsilon})f(x)\leq \sum_{j\in D}y^*_j (f(x\oplus j)-f(x))\]

\begin{lemma}
\label{lem:111} Let $\mathbf{y}(\frac{1}{\delta})$ denote the solution returned from Algorithm \ref{alg:greedy-peak}, we have
$F(\mathbf{y}(\frac{1}{\delta}))\geq \frac{(1-e^{- (1+\frac{2\epsilon n}{1+\epsilon})})(1-\epsilon)}{1+(2n+1)\epsilon}f^+(\mathbf{y}^+)$.
\end{lemma}
\begin{proof}
As proved in \cite{calinescu2011maximizing}, if $g$ is a submodular function, $g^+(\mathbf{y})\leq \min_{S\subseteq \mathcal{S}} (g(S)+\sum_{[vd]\in \mathcal{S}} y_{vd} g_S([vd]))$. Let $\mathbf{y}^+$ denote the optimal solution to problem \textbf{P.B}, assume $\mathbf{y}(t)$ is our solution at round $t$, then for every round $t$, we have
\begin{eqnarray*}&&g^+(\mathbf{y}^+) \leq \min_{S\in \mathcal{S}} (g(S)+\sum_{[vd]\in \mathcal{S}} y^+_{vd} g_S([vd]))\\
&\leq& G(\mathbf{y}(t))+\sum_{[vd]\in \mathcal{S}} y^+_{vd} G_{\mathbf{y}(t)}([vd])\\
&\leq& \frac{1}{1-\epsilon}F(\mathbf{y}(t))\\
&&+\sum_{[vd]\in \mathcal{S}}y^+_{vd} \left(\frac{1}{1-\epsilon}F(\mathbf{y}(t)\oplus \mathbf{1}_{vd})-\frac{1}{1+\epsilon}F(\mathbf{y}(t))\right)
\end{eqnarray*}
It follows that
\begin{eqnarray}
&&\frac{1-\epsilon}{1+\epsilon}f^+(\mathbf{y}^+) \leq (1-\epsilon) g^+(\mathbf{y}^+)\nonumber\\
&\leq& F(\mathbf{y}(t))+\sum_{[vd]\in \mathcal{S}}y^+_{vd} \left(F(\mathbf{y}(t)\oplus \mathbf{1}_{vd})-\frac{1-\epsilon}{1+\epsilon}F(\mathbf{y}(t))\right)\nonumber\\
&=&F(\mathbf{y}(t))+\sum_{[vd]\in \mathcal{S}}y^+_{vd} \left(F(\mathbf{y}(t)\oplus \mathbf{1}_{vd})-F(\mathbf{y}(t))\right)\nonumber\\
&&+\sum_{[vd]\in \mathcal{S}}y^+_{vd} \frac{2\epsilon}{1+\epsilon}F(\mathbf{y}(t))\label{line4}
\end{eqnarray}
The first inequality is due to Lemma \ref{lem:3}. Let $\overline{\mathbf{y}}$ denote the optimal solution to problem \textbf{P.C} in Algorithm \ref{alg:greedy-peak} at round $t$, then we have
\begin{eqnarray}
&&\frac{1-\epsilon}{1+\epsilon}f^+(\mathbf{y}^+)-(1+\sum_{[vd]\in \mathcal{S}}y^+_{vd} \frac{2\epsilon}{1+\epsilon})F(\mathbf{y}(t))\nonumber
\\
&\leq& \sum_{[vd]\in \mathcal{S}}y^+_{vd} F_{\mathbf{y}(t)}([vd])\leq \sum_{[vd]\in \mathcal{S}}\overline{y}_{vd} F_{\mathbf{y}(t)}([vd]) \label{line:2}
\end{eqnarray}
The first inequality is due to (\ref{line4}) and the second inequality is because $\overline{\mathbf{y}}$ is the optimal solution to problem \textbf{P.C}. According to Line \ref{line:1}, the increased value of our solution at round $t+\delta$ is at least
\begin{eqnarray}
&&F(\mathbf{y}(t+\delta))-F(\mathbf{y}(t)) \nonumber\\
&=& \sum_{[vd]\in \mathcal{S}}\delta \overline{y}_{vd} \prod_{[v'd']\neq [vd]}(1-\delta \overline{y}_{v'd'}) F_{\mathbf{y}(t)}([vd])\nonumber\\
&\geq& \sum_{[vd]\in \mathcal{S}}\delta \overline{y}_{vd} (1-\delta)^{mn-1} F_{\mathbf{y}(t)}([vd])\nonumber\\
&=&\delta(1-\delta)^{mn-1}\sum_{[vd]\in \mathcal{S}} \overline{y}_{vd}  F_{\mathbf{y}(t)}([vd])\nonumber\\
&\geq& \delta(1-mn\delta)\sum_{[vd]\in \mathcal{S}} \overline{y}_{vd}  F_{\mathbf{y}(t)}([vd]) \nonumber\\
&=& \delta(1-\frac{1}{mn})\sum_{[vd]\in \mathcal{S}} \overline{y}_{vd}  F_{\mathbf{y}(t)}([vd])\label{line:3}
\end{eqnarray}
The last inequality is due to $\delta=\frac{1}{(mn)^2}$.

(\ref{line:2}) and (\ref{line:3}) together imply that
\begin{eqnarray}
&F(\mathbf{y}(t+\delta))-F(\mathbf{y}(t))\nonumber\\
& \geq\delta(1-\frac{1}{mn})\left(\frac{1-\epsilon}{1+\epsilon}f^+(\mathbf{y}^+)-(1+\sum_{[vd]\in \mathcal{S}} \frac{2\epsilon y^+_{vd}}{1+\epsilon})F(\mathbf{y}(t))\right)\nonumber\\
&\geq \delta(1-\frac{1}{mn})\left(\frac{1-\epsilon}{1+\epsilon}f^+(\mathbf{y}^+)-(1+ \frac{2\epsilon n}{1+\epsilon})F(\mathbf{y}(t))\right)\label{eq:55555}\\
&\geq \delta\left((1-\frac{1}{mn})\frac{1-\epsilon}{1+\epsilon}f^+(\mathbf{y}^+)-(1+ \frac{2\epsilon n}{1+\epsilon})F(\mathbf{y}(t))\right)\label{eq:666}
\end{eqnarray}

Inequality (\ref{eq:55555}) is due to $\sum_{[vd]\in \mathcal{S}} y^+_{vd}\leq n$. Define $\Delta_t=\left((1-\frac{1}{mn})\frac{1-\epsilon}{1+\epsilon}f^+(\mathbf{y}^+)-(1+ \frac{2\epsilon n}{1+\epsilon})F(\mathbf{y}(t))\right)$, according to (\ref{eq:666}), we have $\Delta_{t+\delta}=(1-\delta (1+\frac{2\epsilon n}{1+\epsilon}) )\Delta_{t}$, thus
 \[\Delta_{1/\delta}=(1-\delta (1+\frac{2\epsilon n}{1+\epsilon}) )^{1/\delta}\Delta_{0}=e^{- (1+\frac{2\epsilon n}{1+\epsilon})}\Delta_{0}\]
 It follows that
\[(1+ \frac{2\epsilon n}{1+\epsilon})F(\mathbf{y}(\frac{1}{\delta}))\geq (1-e^{- (1+\frac{2\epsilon n}{1+\epsilon})})(1-\frac{1}{mn})\frac{1-\epsilon}{1+\epsilon}f^+(\mathbf{y}^+)\]\[=(1-e^{- (1+\frac{2\epsilon n}{1+\epsilon})}-o(1))\frac{1-\epsilon}{1+\epsilon}f^+(\mathbf{y}^+)\] ($o(1)$ can be removed when choosing small enough $\delta$ \cite{vondrak2008optimal} )
$\Rightarrow F(\mathbf{y}(\frac{1}{\delta}))\geq \frac{(1-e^{- (1+\frac{2\epsilon n}{1+\epsilon})})(1-\epsilon)}{1+(2n+1)\epsilon}f^+(\mathbf{y}^+)$
\end{proof}

\subsubsection{Rounding} In the rest of this paper, we use $\mathbf{y}$ to denote $\mathbf{y}(\frac{1}{\delta})$ for short. After obtaining $\mathbf{y}$, we use swap rounding \cite{chekuri2010dependent} to round the fractional solution to integral solutions.
%Notice that classic rounding approach relying on the submodularity of the utility function, such as pipage rounding \cite{calinescu2011maximizing}, can not work directly on a general function.

%We first introduce a dummy coupon with value $0$ and define $y_{v0}=1-\sum_{d\in \mathcal{D}}y_{vd}$. For each user $v\in \mathcal{V}$, we assign exactly one discount to her, discount $d\in \mathcal{D}\cup \{0\}$ to user $v$ with probability $y_{vd}$. The returned solution is denoted as $T$. We next prove that $T$ is feasible and its expected cascade $\mathbb{E}[f(T)]$ is close to the optimal solution.

\begin{theorem}
\label{thm:1}
Our rounding approach returns a feasible solution $T$ to problem \textbf{P.A}, and
\[\mathbb{E}[f(T)]\geq \beta f^+(\mathbf{y}^+)\] where $\beta= \left(\frac{1-\epsilon}{1+\epsilon}\right) \frac{(1-e^{- (1+\frac{2\epsilon n}{1+\epsilon})})(1-\epsilon)}{1+(2n+1)\epsilon}$.
\end{theorem}
\begin{proof}
We first prove the feasibility of our rounding approach. First, our policy trivially satisfied Cc in $\textbf{P.A.1}$. Since the solution returned from swap rounding satisfies matroid constraint, our solution satisfies Ca in $\textbf{P.A.1}$. Because each u-c pair is selected with probability $y_{vd}$ and $\mathbf{y}$ is a feasible solution to problem \textbf{P.B}, the expected cost of $T$ is $\sum_{d\in \mathcal{D}} y_{vd}p_v(d)d\leq B$ (Cb in $\textbf{P.A}$ is satisfied).

We next prove that $\mathbb{E}[f(T)]\geq \left(\frac{1-\epsilon}{1+\epsilon}\right) F(\mathbf{y})$. First, because $g$ is a submodular function, we have $\mathbb{E}[g(T)]\geq  G(\mathbf{y})$. Because $F(\mathbf{y})\leq (1+\epsilon)G(\mathbf{y})$, we have
\begin{equation} F(\mathbf{y})\leq  (1+\epsilon)\mathbb{E}[g(T)]\leq \frac{1+\epsilon}{1-\epsilon} \mathbb{E}[f(T)]
\end{equation}

Therefore,  we have
\begin{eqnarray}\mathbb{E}[f(T)]&\geq& \left(\frac{1-\epsilon}{1+\epsilon}\right) F(\mathbf{y}) \label{eq:4444}\\
&\geq&  \left(\frac{1-\epsilon}{1+\epsilon}\right) \frac{(1-e^{- (1+\frac{2\epsilon n}{1+\epsilon})})(1-\epsilon)}{1+(2n+1)\epsilon}f^+(\mathbf{y}^+)\nonumber
\end{eqnarray}
The second inequality is due to  Lemma \ref{lem:111}.
\end{proof}

\section{Extension: Incorporating Distribution Cost}
\label{sec:extended}
We now discuss one extension of our analysis thus far. The difference between the extended model and the previous model is that in
the new model, we add one more constraint to Definition \ref{def:0}. %In particular, we assume a feasible policy can not allocate coupons to more than $K$ users.
This additional constraint models the fact that the company only has limited resource to produce and distribute coupons to different users.
Assume the cost of allocating a coupon to user $v$ is $a_v$, we say a policy $\Theta$ is feasible if and only the following conditions are satisfied:
\begin{itemize}
\item(Attention Constraint) $\forall \theta_S>0, \forall v: |\mathcal{S}_v\cap S|\leq 1$.
\item(Coupon Cost Constraint) $ \sum_{S\in \mathcal{S}}\theta_{S} c(S)\leq B$.
\item(Coupon Distribution Budget Constraint) $\forall \theta_S>0: \sum_{[vd]\in S} a_v\leq K$.
\end{itemize}
%\end{definition}

Similar to the previous model, we formulate our problem $\textbf{P.A.1}$  as follows.

 \begin{center}
\framebox[0.45\textwidth][c]{
\enspace
\begin{minipage}[t]{0.45\textwidth}
\small
$\textbf{P.A.1}$ $\max_{\Theta} \sum_{S\in \mathcal{S}}\theta_{S} f(S)$\\
\textbf{subject to:}
\begin{equation*}
\begin{cases}
\forall \theta_S>0, \forall v: |\mathcal{S}_v\cap S|\leq 1 \mbox{ (Ca)}\\
\sum_{S\in \mathcal{S}}\theta_{S} c(S)\leq B \mbox{ (Cb)}\\
\sum_{S\in \mathcal{S}}\theta_{S}= 1 \mbox{ (Cc)}\\
\forall \theta_S>0: \sum_{[vd]\in S} a_v\leq K \mbox{ (Cd)}
\end{cases}
\end{equation*}
\end{minipage}
}
\end{center}
\vspace{0.1in}

Its relaxation $\textbf{P.B.1}$ can be formulated as follows.
\begin{center}
\framebox[0.45\textwidth][c]{
\enspace
\begin{minipage}[t]{0.45\textwidth}
\small
$\textbf{P.B.1}$
\emph{Maximize $f^+(\mathbf{y})$}\\
\textbf{subject to:}
\begin{equation*}
\begin{cases}
%$\forall \theta_{\mathcal{S}}> 0: |\mathcal{S}|\leq 1 $\\
%\forall v\in \mathcal{V}: \sum_{d\in D} y_{vd}\leq1 \quad(C1)\\
%  \sum_{[vd]\in \mathcal{S}} y_{vd}p_v(d)d\leq B \quad(C2)\\
%\forall [vd]\in \mathcal{S}: y_{vd} \in[0,1] \quad (C3)\\
\mbox{(C1)$\sim$ (C3)}\\
 \sum_{[vd]\in \mathcal{S}} a_v y_{vd}\leq K 
\end{cases}
\end{equation*}
\end{minipage}
}
\end{center}
\vspace{0.1in}

%\subsection{An Approximate Solution}
 We follow a similar idea used in the previous section to design our algorithm, e.g., compute a fractional solution first and then round it to integral solution. However, to tackle the new challenge brought by the additional constraint, we need a brand new solution to ensure the feasibility of the final solution.

We first introduce a new problem $\textbf{P.A.2}$ as follows.  $\textbf{P.A.2}$ is a restricted version of $\textbf{P.A.1}$ where the coupon distribution budget $K$ is scaled down by a factor of $b<1$.

 \begin{center}
\framebox[0.45\textwidth][c]{
\enspace
\begin{minipage}[t]{0.45\textwidth}
\small
$\textbf{P.A.2}$ $\max_{\Theta} \sum_{S\in \mathcal{S}}\theta_{S} f(S)$\\
\textbf{subject to:}
\begin{equation*}
\begin{cases}
\forall \theta_S>0, \forall v: |\mathcal{S}_v\cap S|\leq 1\\
\sum_{S\in \mathcal{S}}\theta_{S} c(S)\leq B\\
\sum_{S\in \mathcal{S}}\theta_{S}= 1\\
\forall \theta_S>0: \sum_{[vd]\in S} a_v\leq b K
\end{cases}
\end{equation*}
\end{minipage}
}
\end{center}
\vspace{0.1in}

We next introduce a relaxed version of problem $\textbf{P.A.2}$.

 \begin{center}
\framebox[0.45\textwidth][c]{
\enspace
\begin{minipage}[t]{0.45\textwidth}
\small
$\textbf{P.B.2}$
\emph{Maximize $f^+(\mathbf{y})$}\\
\textbf{subject to:}
\begin{equation*}
\begin{cases}
%$\forall \theta_{\mathcal{S}}> 0: |\mathcal{S}|\leq 1 $\\
%\forall v\in \mathcal{V}: \sum_{d\in D} y_{vd}\leq1 \quad(C1)\\
%  \sum_{[vd]\in \mathcal{S}} y_{vd}p_v(d)d\leq B \quad(C2)\\
%\forall [vd]\in \mathcal{S}: y_{vd} \in[0,1] \quad (C3)\\
\mbox{(C1)$\sim$ (C3)}\\
 \sum_{[vd]\in \mathcal{S}} a_v y_{vd}\leq bK \quad\mbox{(C4)}
\end{cases}
\end{equation*}
\end{minipage}
}
\end{center}
\vspace{0.1in}
\subsubsection{Continuous Greedy Algorithm}
We present the continuous greedy algorithm in Algorithm \ref{alg:greedy-peak1}. The only difference between Algorithm \ref{alg:greedy-peak1} and Algorithm \ref{alg:greedy-peak} is  that in  Algorithm \ref{alg:greedy-peak1},  we replace $\textbf{P.C}$ by $\textbf{P.C.1}$ to incorporate constraint (C4).

\begin{algorithm}[h]
{\small
\caption{Continuous Greedy}
\label{alg:greedy-peak1}
%\textbf{Input:} Social network $\mathcal{G}$, budget $\mathcal{B}$, individual attention constraint $\kappa_i$, overall attention constraint $K$.\\
%\textbf{Output:} Seed set $\mathcal{S}$.
\begin{algorithmic}[1]
%\STATE Set $\delta=1/(nm)^2, t=0, f(\mathbf{0})=0$.
%\WHILE{$t<1$}
%\STATE Let $S(t)$ contain each $[vd]$ independent with probability $y_{vd}(t)$.
%\STATE For each $[vd]\in \mathcal{S}$, estimate $\omega_{vd}=\mathbb{E}[f_{S(t)}([vd])]$.
%\STATE Use LP solver to solve the following problem and obtain the optimal solution $\overline{\mathbf{y}}$
\STATE Replace $\textbf{P.C}$ in Algorithm \ref{alg:greedy-peak} by $\textbf{P.C.1}$
\\
\framebox[0.42\textwidth][c]{
\enspace
\begin{minipage}[t]{0.42\textwidth}
\small
$\textbf{P.C.1}$:
\emph{Maximize $\sum_{[vd]\in \mathcal{S}}\omega_{vd}y_{vd}$ }\\
\textbf{subject to:} Conditions (C1)$\sim$(C4)
\end{minipage}
}
%\vspace{0.1in}
%\STATE Let $y_{vd}(t+\delta)=y_{vd}(t)+\overline{y}_{vd}$;
%\STATE Increment $t=t+\delta$;
%\ENDWHILE
\end{algorithmic}
}
\end{algorithm}

\subsubsection{Rounding} %In the rest of this paper, we use $\mathbf{y}$ to denote $\mathbf{y}(\frac{1}{\delta})$ for short. After obtaining $\mathbf{y}$, we next round the fraction solution to integral solutions.
 We next design a brand new rounding approach that satisfies all constraints.  In the rest of this paper, we use $\hat{\mathbf{y}}$ to denote the solution returned from Algorithm \ref{alg:greedy-peak1}, let $\hat{\mathbf{y}}^+$ denote the optimal solution to problem $\textbf{P.B.2}$ and $\mathbf{y}^+$ denote the optimal solution to problem $\textbf{P.B.1}$. Our approach consists of two major parts: a rounding
stage (Step 1) and a conflict resolution stage (Step 2).

\emph{Step 1:} Given $\hat{\mathbf{y}}$, we first introduce a dummy coupon with value $0$ and for each user $v\in \mathcal{V}$, define $\hat{y}_{v0}=1-\sum_{d\in \mathcal{D}}\hat{y}_{vd}$. For each user $v\in \mathcal{V}$, we assign exactly one coupon to her, coupon $d\in \mathcal{D}\cup \{0\}$ to user $v$ with probability $\hat{y}_{vd}$. The returned solution is denoted as $I$. We notice that $I$ may not be a feasible solution, e.g., the coupon distribution cost of $I$ could be larger than $K$.

\emph{Step 2:} Consider u-c pairs from $I$ in non-decreasing order for their distribution cost, let the current u-c pair survive if it does not violate C4, else we discard it. We add all survived u-c pairs to $T$.

We next prove that $T$ is feasible and its expected cascade $\mathbb{E}[f(T)]$ is close to the optimal solution.

 \begin{theorem}
 \label{thm:2}
$T$ is feasible and \[\mathbb{E}[f(T)]\geq (1-2b)b\beta  f^+(\mathbf{y}^+)\]  where $b\in (0,1/2]$ and $\beta= \left(\frac{1-\epsilon}{1+\epsilon}\right) \frac{(1-e^{- (1+\frac{2\epsilon n}{1+\epsilon})})(1-\epsilon)}{1+(2n+1)\epsilon}$.
 \end{theorem}
 \begin{proof} We first prove the feasibility of $T$. First, our policy trivially satisfied Cc in $\textbf{P.A.1}$. Consider the random set $I$ returned from Step 1, it allocates at most one coupon to each user (Ca in $\textbf{P.A.1}$ is satisfied) and the expected cost of $T$ is $\sum_{d\in \mathcal{D}} y_{vd}p_v(d)d\leq B$ (Cb in $\textbf{P.A.1}$ is satisfied). Since $T$ is a subset of $I$, it also satisfied Ca and Cb. Moreover, according to Step 2, $T$  satisfied Cd in $\textbf{P.A.1}$.

We next prove $\mathbb{E}[f(T)]\geq (1-2b)b\beta  f^+(\mathbf{y}^+)$ by putting together the following three lemmas.

\begin{lemma}
\label{lem:1112}
$f^+(\mathbf{\hat{y}}^+) \geq b f^+(\mathbf{y}^+)$.
\end{lemma}
 Proof of Lemma \ref{lem:1112}.  First, it is easy to verify that $b \mathbf{y}^+$ is a feasible solution to $\textbf{P.B.2}$, thus we have $f^+(\mathbf{\hat{y}}^+) \geq f^+(b \mathbf{y}^+)$. Assume the value of $f^+(\mathbf{y}^+)$ is achieved at $\alpha_S^*$, we have $b f^+(\mathbf{y}^+)=\sum_{S\in \mathcal{S}}  \alpha_S^* f(S)\leq f^+(b \mathbf{y}^+)\leq f^+(\mathbf{\hat{y}}^+)$.

\begin{lemma}
\label{lem:11122}
$\mathbb{E}[f(I)]\geq \beta f^+(\mathbf{\hat{y}}^+)$.
\end{lemma}
Proof of Lemma \ref{lem:11122}.
It was shown in \cite{vondrak2008optimal} that ``Step 1'' can be used to round solutions in the basic partition matroid polytope without
losing in terms of the objective function. Because $g$ is a submodular function, we have $\mathbb{E}[g(I)]\geq  G(\mathbf{\hat{y}})$. The rest of the proof is now analogous to the proof of Theorem \ref{thm:1}.

\begin{lemma}
\label{lem:111222}
$\mathbb{E}[f(T)]\geq (1-2b) \mathbb{E}[f(I)]$.
\end{lemma}
Proof of Lemma \ref{lem:111222}.
According to Step 1, given the fractional solution $\mathbf{\hat{y}}$, we decide the coupon allocation for each user independently, that is the coupon allocated to each user is independent from each other. This nice property enables us to follow a similar proof in \cite{vondrak2011submodular} to prove that for each $[vd]\in \mathcal{S}$, $\Pr[[vd]\in T|[vd]\in I]\geq 1-2b$, that is for any coupon that is having survived after Step 1, the probability that it still survives after Step 2 is at least $1-2b$. According to Theorem 5 in \cite{bansal2010k}, we have $\mathbb{E}[f(T)]\geq (1-2b) \mathbb{E}[f(I)]$.

\end{proof}

We emphasize that Theorem \ref{thm:2} is of independent interest and may find applications in any approximate submodular maximization problem subject to knapsack and a basic partition matroid constraints.

\section{Conclusion}
In this paper, we study coupon allocation problem in social networks. Our framework allows a general utility function and more complicated constraints. Therefore, existing techniques relying on the submodularity of the utility function can not apply to our problem directly. We propose a novel approximate algorithm with approximation ratio depending on $\epsilon$. Although we limit our attention to coupon allocation problem in this paper, our results apply to a broad range of approximate submodular maximization problems.
\bibliographystyle{named}
\bibliography{reference}

\end{document}